\documentclass[12pt,a4paper]{article}
\usepackage{epsfig}
\usepackage{amsfonts}
\usepackage{amsopn}
\usepackage{amsmath,amssymb}
\usepackage{verbatim,amsthm}

\DeclareMathOperator{\diag}{diag}

\DeclareMathOperator{\supp}{supp}
\DeclareMathOperator{\tr}{tr}

\DeclareMathAlphabet{\mathpzc}{OT1}{pzc}{m}{it}

\theoremstyle{plain}
\newtheorem{thm}{Theorem}

\newtheorem{lem}[thm]{Lemma}

\newtheorem{prop}[thm]{Proposition}

\theoremstyle{remark}

\title{Weighted Supermembrane Toy Model}

\author{Douglas Lundholm\thanks{e-mail: dogge@math.kth.se}}

\date{\scriptsize{Department of Mathematics, Royal Institute of Technology\\SE-100 44 \ Stockholm, Sweden}}

\begin{document}

\maketitle

\begin{abstract}
	A weighted Hilbert space approach to the study of 
	zero-energy states of supersymmetric matrix models is introduced.
	Applied to a related but technically simpler model, it is shown that the 
	spectrum of the corresponding weighted Hamiltonian 
	simplifies to become purely discrete for sufficient weights.
	This follows from a bound for the number of negative eigenvalues 
	of an associated matrix-valued Schr\"odinger operator.
\end{abstract}

{	\small
}

\setlength\arraycolsep{2pt}
\def\arraystretch{1.5}

\section{Introduction}

There are many difficulties in the study of 
zero-energy states of supersymmetric matrix models.
Some arise due to the fact that the spectrum 
of the associated Hamiltonian is continuous\footnote{or
purely essential, to be precise}, starting at zero.
Accordingly, we expect it to be useful to shift
to a weighted Hilbert space on which the 
corresponding operator has a discrete spectrum.

In this work we illustrate the applicability of the technique
to a simplified model which, 
despite its technical simplicity,
still shares many of the features (and difficulties)
of the original matrix models. The spectral properties of this
so-called \emph{supermembrane toy model}, and its purely bosonic counterpart, 
has been previously studied in 
\cite{Hoppe,Simon,Simon_asymptotics,dWLN,Koubek,Hasler,Graf-Hasler-Hoppe,Korcyl},
while the underlying geometry of the model was emphasized in \cite{geosusy}.
Our results on the technically much more complicated matrix models
will be presented in a forthcoming paper.

In Section 2 we recall the formulation of the original toy model and introduce 
the weighted Hilbert space approach.
In Sections 3 and 4 we investigate the spectral properties of the
weighted model, and show that, under a certain condition on the
parameter of the weight, the spectrum of the weighted model is in
fact discrete. This is accomplished using a Cwikel-Lieb-Rozenblum-type bound 
for operator-valued potentials which is derived in Section 5 from a result in \cite{Hundertmark}.
We also note that the technique we use could provide
a geometric understanding of the eigenvalue asymptotics
of the purely bosonic model.

\section{The original and weighted models}

The supermembrane toy model,
also called the supersymmetric $x^2y^2$ potential, 
is defined by the Hamiltonian operator
\begin{equation} \label{toy_hamiltonian}
	H = -\Delta + V + H_F = -\partial_x^2 - \partial_y^2 + x^2y^2 + x\gamma_1 - y\gamma_2
\end{equation}
(where $\gamma_k$ are Pauli matrices),
acting on the Hilbert space 
$$
	\mathcal{H} = L^2(\mathbb{R}^2,dxdy) \otimes \mathbb{C}^2.
$$
A corresponding hermitian supercharge operator is given by
$$
	Q = -i(\partial_x\gamma_1 + \partial_y\gamma_2) + xy\gamma_3,
$$
such that $Q^2 = H \ge 0$.

The matrix-valued Schr\"odinger operator
$H$ can be formally defined as a self-adjoint operator 
through the closure of the quadratic form 
corresponding to the expression \eqref{toy_hamiltonian} on 
$C_0^\infty(\mathbb{R}^2) \otimes \mathbb{C}^2$.
The spectrum of $H$ is $\sigma(H) = [0,\infty)$ 
due to potential valleys along the coordinate axes, 
where the lower eigenvalue of the matrix potential,
$$
	V + (H_F)_- = x^2y^2 - \sqrt{x^2 + y^2},
$$
tends to negative infinity
-- precisely cancelling the localization energy due to
the narrowing of the valley.

We would like to make the spectrum of the model discrete by introducing 
a weighted Hilbert space. We define
\begin{equation*}
	\mathcal{H}_w := L^2(\mathbb{R}^2,\rho(x,y)dxdy) \otimes \mathbb{C}^2, \qquad
	\rho(x,y) := (1 + x^2 + y^2)^{-\frac{\alpha}{2}},
\end{equation*}
with $\alpha \ge 0$. The inner product on $\mathcal{H}_w$ is then given by
$$
	\langle \Phi, \Psi \rangle_w = \langle \Phi, \rho\Psi \rangle 
	= \int_{\mathbb{R}^2} \frac{ \langle \Phi(x,y),\Psi(x,y) \rangle_{\mathbb{C}^2} }{(1 + x^2 + y^2)^{\frac{\alpha}{2}}} dxdy.
$$
The corresponding self-adjoint operator $\tilde{H}$ on $\mathcal{H}_w$ 
is defined through the same quadratic form on $C_0^\infty(\mathbb{R}^2) \otimes \mathbb{C}^2$;
$$
	\langle \Psi, \tilde{H} \Psi \rangle_w := \langle \Psi, H \Psi \rangle = \|Q \Psi\|^2 \ge 0.
$$
It follows that, if we define $\tilde{Q} := \rho^{-\frac{1}{2}}Q$
(with adjoint w.r.t $\mathcal{H}_w$ given by $\tilde{Q}^* = \rho^{-1}Q\rho^{\frac{1}{2}}$)
we have
$$
	\langle \Psi, \tilde{H} \Psi \rangle_w = \| \tilde{Q}\Psi \|_w^2.
$$
We observe in general that any solution of $H\Psi = 0$ in $\mathcal{H}$
is also a solution of $\tilde{H}\Psi = 0$ in the weighted Hilbert space. 
On the other hand, finding a solution 
of $\tilde{H}\Psi = 0$ in $\mathcal{H}_w$ does
yield a (smooth\footnote{by elliptic regularity}) 
solution to the differential equation $H\Psi = 0$,
but its decay rate may be insufficient for square-integrability.
For this particular model it is known \cite{Graf-Hasler-Hoppe}
that there is \emph{no} solution in $\mathcal{H}$.

\section{Spectrum of $\tilde{H}$ for $\alpha < 2$}

Continuity of the spectrum of $H$ can be proved (see \cite{dWLN}) by,
for any $\mu \ge 0$, finding a Weyl sequence
$(\Psi_t)$ in $\mathcal{H}$ such that $\|\Psi_t\| = 1\ \forall t$ and
$$
	\langle \Psi_t, (H-\mu)^2 \Psi_t \rangle \to 0, \quad t \to \infty.
$$
Explicitly (and for $\mu = 0$ for simplicity), we take
$$
	\Psi_t(x,y) := \chi_t(x) \phi_x(y) \xi,
$$
where $\chi_t$ is a cut-off function s.t.
$\chi_t(x) := t^{-\frac{1}{2}} \chi(x/t)$, $\chi \in C_0^\infty[1,2]$, $\int_\mathbb{R} \chi^2 = 1$,
$\phi_x$ is the normalized groundstate of the harmonic oscillator $-\partial_y^2 + x^2y^2$,
$$
	\phi_x(y) := \left(\frac{x}{\pi}\right)^{\frac{1}{4}} e^{-\frac{1}{2}xy^2},
$$
and $\xi \in \mathbb{C}^2$ is a unit eigenvector,
$
	\gamma_1 \xi = -\xi.
$
One finds that $\|\Psi_t\|^2 = 1$ and
$$
	\langle \Psi_t, H \Psi_t \rangle 
	\ \le \ \int |\chi_t \chi_t''| dx + c_1 \int x^{-1} |\chi_t \chi_t'| dx + c_2 \int x^{-2} |\chi_t|^2 dx
	\ \le \ c t^{-2}
$$
(here, and in the following, $c,c_1$, etc. denote some positive constants).

Taking the same sequence for the weighted case, we also find
$$
	\langle \Psi_t, \tilde{H} \Psi_t \rangle_w = \langle \Psi_t, H \Psi_t \rangle \le c t^{-2}.
$$
However, the norm is now
\begin{eqnarray*}
	\|\Psi_t\|_w^2 
	&=& \int_x |\chi_t(x)|^2 \int_y \frac{\phi_x(y)^2}{(1 + x^2 + y^2)^{\frac{\alpha}{2}}} dy dx \\
	&\ge& c_1 \int_x \frac{|\chi_t(x)|^2}{(1 + 4t^2 + c_2t^{-1})^{\frac{\alpha}{2}}} dx
	\ge c_3 (1 + 4t^2)^{-\frac{\alpha}{2}},
\end{eqnarray*}
so that, for $\alpha < 2$,
$$
	\frac{ \langle \Psi_t, \tilde{H} \Psi_t \rangle_w }{ \|\Psi_t\|_w^2 } 
	\le c (1+4t^2)^{\frac{\alpha}{2}} t^{-2}
	\to 0, \quad t \to \infty.
$$
Thus, $\Psi_t$ still approximates a zero-energy eigenfunction, 
but since its support moves out to infinity, this indicates
that the spectrum of $\tilde{H}$ is still continuous for $0 \le \alpha < 2$.

\section{Spectrum of $\tilde{H}$ for $\alpha > 2$}

The spectrum of $\tilde{H}$ is discrete if and only if, for all $\lambda > 0$,
the rank of the spectral projection of $\tilde{H}$ on $(-\infty,\lambda)$ is finite.
Equivalently, if and only if
$$
	\sup_{W_\lambda} \dim W_\lambda < \infty,
$$
where $W_\lambda$ are subspaces of $C_0^\infty(\mathbb{R}^2) \otimes \mathbb{C}^2$ such that,
for all $\Psi \in W_\lambda$,
\begin{equation} \label{disc_condition}
	\langle \Psi, \tilde{H} \Psi \rangle_w < \lambda \|\Psi\|_w^2.
\end{equation}
Note that \eqref{disc_condition} is equivalent to
$$
	\langle \Psi, (H - \lambda\rho) \Psi \rangle < 0.
$$
It follows that the spectrum of $\tilde{H}$ is discrete if and only if 
the operator $H - \lambda\rho$, on the original Hilbert space $\mathcal{H}$, 
has finitely many negative eigenvalues for any $\lambda > 0$, 
or more precisely, $N(H - \lambda\rho) < \infty$,
where we denote by $N(A)$ the (possibly infinite) rank of the spectral 
projection on $(-\infty,0)$ of a self-adjoint operator $A$.
We will prove the following theorem.

\begin{thm} \label{thm_H_lambda}
	For all $\lambda > 0$ and $\alpha > 2$, the operator
	\begin{equation} \label{H_lambda}
		H_\lambda := H - \lambda\rho = -\partial_x^2 - \partial_y^2 + x^2y^2 + x\gamma_1 - y\gamma_2 - \frac{\lambda}{(1 + x^2 + y^2)^{\frac{\alpha}{2}}}
	\end{equation}
	has finitely many negative eigenvalues.
	Furthermore, the number of negative eigenvalues is bounded by
	\begin{equation} \label{H_lambda_bound}
		N(H_\lambda) \le 
			C(\alpha) + \frac{2^{12}\pi C_3}{27(\alpha-2)^3} \lambda^{\frac{3}{2} - \epsilon(\alpha)},
	\end{equation}
	where $C(\alpha)$ and $C_3$ are positive constants, 
	and $0 < \epsilon(\alpha) < \frac{1}{2}(\alpha-2)$.
\end{thm}

\noindent
Our strategy is to prove 
this by splitting 
the domain of the operator into
different regions -- based on the geometry of the potential valleys --
and introducing 
Dirichlet boundary conditions between these regions
by means of a partition of unity.
The unbounded region along the potential valley is then 
shown to admit only finitely many negative eigenvalues
using a Cwikel-Lieb-Rozenblum bound for operator-valued potentials.
In order to illustrate the latter part of
this procedure, let us first prove that $H_\lambda$
defined on the region $x>1$ 
and with Dirichlet boundary condition at $x=1$ 
has finitely many eigenvalues below zero when $\alpha > 2$.
However, despite the fact that there is a reflection symmetry between $x$ and $y$,
this result cannot be directly applied to prove Theorem \ref{thm_H_lambda}
because of inconvenient intersections between regions of this form.
Instead, we will introduce a different set of coordinates and 
define the regions with respect to those.

\subsection{Cartesian coordinates}

In the cartesian coordinates $(x,y)$ we consider the region
$\Omega := (1,\infty) \times \mathbb{R}$
and the semi-bounded operator $H_\lambda^{xy}$ defined by closure of the
quadratic form corresponding to \eqref{H_lambda} on 
$C_0^\infty(\Omega) \otimes \mathbb{C}^2$.
Note that, for $\Psi \in C_0^\infty(\Omega) \otimes \mathbb{C}^2$,
\begin{eqnarray*}
	\lefteqn{ \int_\Omega \langle \Psi, H_\lambda^{xy} \Psi \rangle_{\mathbb{C}^2} dxdy }\\
	&\ge& \int_\Omega \left\langle \Psi, \left(
	 	-\partial_x^2 -\partial_y^2 + x^2 \left( y - \frac{1}{2x^2}\gamma_2 \right)^2 - \frac{1}{4x^2} - x - \frac{\lambda}{x^{\alpha}}
		\right) \Psi \right\rangle_{\mathbb{C}^2} dxdy.
\end{eqnarray*}
Choosing the representation $\gamma_2 = \diag(1,-1)$, 
and making separate coordinate transformations 
$\tilde{y} := y \pm \frac{1}{2x^2}$
in the integrals over the corresponding components of $\Psi$, 
we find
$$
	H_\lambda^{xy} 
	\ge -\partial_x^2 - \frac{1}{4x^2} -\partial_{\tilde{y}}^2 + x^2 \tilde{y}^2 - x - \frac{\lambda}{x^{\alpha}},
$$
where the two-dimensional scalar Schr\"odinger operator on the
r.h.s. can be considered as a Schr\"odinger operator on the
interval $(1,\infty)$ with an operator-valued potential
$V(x) = -\partial_{\tilde{y}}^2 + x^2 \tilde{y}^2 - x - \lambda x^{-\alpha}$
acting on $L^2(\mathbb{R},d\tilde{y})$.
This shifted harmonic oscillator, 
with the projection onto its $k$:th eigenvector denoted $P_k$,
is bounded below by its negative part 
$$
	V(x)_- = \sum_{k=0}^\infty (2kx - \lambda x^{-\alpha})_- P_k
	\ge -\lambda x^{-\alpha} \sum_{0 \le k \le \lambda/2} P_k.
$$
Applying Lemma \ref{lem_integral_crit_bound} below
(with a factor 2 coming from the trace over $\mathbb{C}^2$), we find
\begin{eqnarray*}
	N(H_\lambda^{xy}) 
	&\le& N\left( \left(-\partial_x^2 - \frac{1}{4x^2}\right) - \frac{\lambda}{x^{\alpha}} \sum_{0 \le k \le \lambda/2} P_k \right) \\
	&\le& 8\pi C_3 \int_1^\infty (1 + \lambda/2)\left( \lambda x^{-\alpha} \right)^{\frac{3}{2}} x^2 (\ln x)^2 \thinspace dx,
\end{eqnarray*}
which is finite for $\alpha > 2$.

\subsection{Parabolic coordinates}

Consider the coordinate transformation (cp. \cite{geosusy})
$$
	(x,y) \mapsto (u,v) := \left( \frac{1}{2}(x^2-y^2), xy \right),
$$
which is conformal everywhere except at the origin, 
and maps e.g. the open right half-plane bijectively onto 
the whole plane with the negative real line removed.
We introduce the regions (see Figure \ref{fig_regions})
\begin{figure}[t]
	\centering
	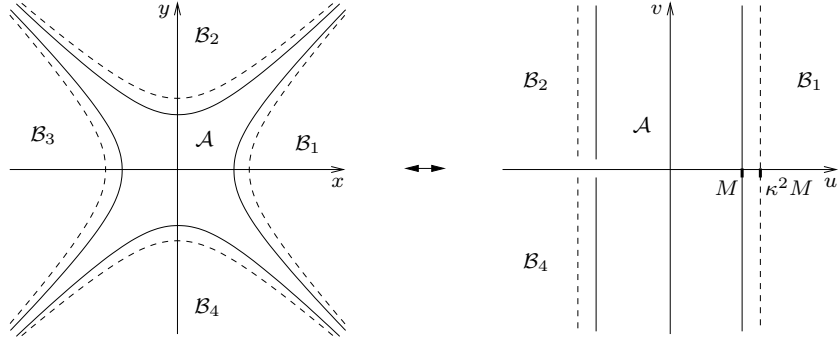
	\caption{Partition of $\mathbb{R}^2$ into the regions 
		$\mathcal{A},\mathcal{B}_1,\mathcal{B}_2,\mathcal{B}_3,\mathcal{B}_4$.}
	\label{fig_regions}
\end{figure}
\begin{eqnarray*}
	\mathcal{A} 	&:& -M < u < M, \\
	\mathcal{B}_1	&:& u > M, x>0,
\end{eqnarray*}
and the corresponding reflections $\mathcal{B}_{2,3,4}$ 
of $\mathcal{B}_1$ in the 
symmetry lines $x=0$ and $x=y$,
together with their union 
$\mathcal{B} := \cup_{j=1}^4 \mathcal{B}_j = \mathbb{R}^2 \setminus \bar{\mathcal{A}}$.
We will also make use of rescaled versions of these regions, 
e.g. $\kappa \mathcal{A}$, with a fixed $\kappa > 1$.

Take a partition of unity, $1 = \chi_{\mathcal{A}}^2 + \chi_{\mathcal{B}}^2$, 
such that $\chi_{\mathcal{A},\mathcal{B}} \in C^{\infty}(\mathbb{R}^2;[0,1])$,
$\chi_{\mathcal{A}} = 1$ on $\mathcal{A}$, and 
$\chi_{\mathcal{B}} = 1$ on $\kappa \mathcal{B}$.
It follows that, for any $\Psi \in C_0^\infty(\mathbb{R}^2) \otimes \mathbb{C}^2$,
\begin{equation} \label{quad_form_split}
	\langle \Psi, H_\lambda \Psi \rangle 
	= \langle \Psi, H_\lambda (\chi_{\mathcal{A}}^2 + \chi_{\mathcal{B}}^2)\Psi \rangle
	= \langle \chi_{\mathcal{A}}\Psi, H_\lambda^{\mathcal{A}} \chi_{\mathcal{A}}\Psi \rangle
	+ \langle \chi_{\mathcal{B}}\Psi, H_\lambda^{\mathcal{B}} \chi_{\mathcal{B}}\Psi \rangle,
\end{equation}
where
\begin{equation} \label{H_A_B}
	H_\lambda^{\mathcal{A},\mathcal{B}} := H_\lambda - |\nabla \chi_{\mathcal{A}}|^2 - |\nabla \chi_{\mathcal{B}}|^2
\end{equation}
denotes the corresponding operator restricted to the domain
$\kappa \mathcal{A}$ resp. $\mathcal{B}$ with Dirichlet 
boundary condition at the boundary $|u| = \kappa^2 M$ resp. $|u| = M$.
As we will see, the additional negative potential terms in \eqref{H_A_B},
denoted $-V_\chi$, will not cause any problems because
they are supported on a region $\kappa\mathcal{A} \cap \mathcal{B}$
where the potential tends rapidly to infinity.
Using that $N(A+B) \le N(A) + N(B)$ for any two self-adjoint operators $A,B$,
we obtain from the quadratic form expression \eqref{quad_form_split} that
$$
	N(H_\lambda) \le N(H_\lambda^{\mathcal{A}}) + N(H_\lambda^{\mathcal{B}})
	= N(H_\lambda^{\mathcal{A}}) + \sum_{j=1}^4 N(H_\lambda^{\mathcal{B}_j}).
$$

Consider first the region $\mathcal{B}_1$.
Under the coordinate transformation, we find (cp. \cite{geosusy})
$
	\Delta_{xy} = h^{-2}\Delta_{uv}
$
and $dxdy = h^2dudv$,
with scale factor 
$
	h = (x^2+y^2)^{-\frac{1}{2}} = 2^{-\frac{1}{2}}(u^2+v^2)^{-\frac{1}{4}},
$
so that for any $\Psi \in C_0^\infty(\mathcal{B}_1) \otimes \mathbb{C}^2$
\begin{eqnarray} \label{quad_form_transformation} 
	\lefteqn{ \int_{\mathcal{B}_1} \langle \Psi, (-\Delta_{xy} + x^2y^2 + x\gamma_1 - y\gamma_2 - \lambda\rho - V_\chi) \Psi \rangle_{\mathbb{C}^2} \thinspace dxdy } \nonumber \\
	&=& \int_{\mathcal{B}_1} \langle \Psi, (-\Delta_{uv} + h^2v^2 + h\gamma_u - \lambda h^2\rho - V_\chi^{uv}) \Psi \rangle_{\mathbb{C}^2} \thinspace dudv \nonumber \\
	&\ge& \int_{u=M}^{\infty} \int_{v=-\infty}^{\infty} \Big\langle \Psi, 
		\bigg(-\partial_u^2 -\partial_v^2 + \frac{v^2}{2\sqrt{u^2+v^2}} - \frac{1}{\sqrt{2}(u^2+v^2)^{\frac{1}{4}}} \\
	&&\hspace{3.45cm}  - \frac{\lambda}{2\sqrt{u^2+v^2}(1+2\sqrt{u^2+v^2})^{\frac{\alpha}{2}}}  - V_\chi^{uv} \bigg) 
		\Psi \Big\rangle_{\mathbb{C}^2} dvdu,  \nonumber
\end{eqnarray}
where $\gamma_u := h(x\gamma_1-y\gamma_2)$, so that $\gamma_u^2=1$.
We have also used that
$$
	h^2 V_\chi = (h|\nabla_{xy}\chi_{\mathcal{A}}|)^2 + (h|\nabla_{xy}\chi_{\mathcal{B}}|)^2 
	= |\nabla_{uv}\chi_{\mathcal{A}}|^2 + |\nabla_{uv}\chi_{\mathcal{B}}|^2
	=: V_\chi^{uv},
$$
which (with a suitably chosen $\chi_\mathcal{A}$)
is independent of $v$, bounded by $c_1/M^2$,
and supported on $M \le |u| \le \kappa^2 M$.

Let us think of the resulting scalar Schr\"odinger 
operator in the r.h.s. of \eqref{quad_form_transformation}, 
call it $H^{uv}_\lambda$,
as acting on $L^2([M,\infty),du) \otimes \mathpzc{h}$ 
with fiber $\mathpzc{h} = L^2(\mathbb{R},dv)$.
Let
\begin{equation} \label{H_u}
	H_u := -\partial_v^2 + \frac{v^2}{2(u^2+v^2)^{\frac{1}{2}}} - \frac{1}{\sqrt{2}(u^2+v^2)^{\frac{1}{4}}},
\end{equation}
denote part of the one-dimensional Schr\"odinger operator acting on $\mathpzc{h}$,
and observe that
\begin{equation} \label{H_uv_estimate}
	H^{uv}_\lambda 
	\ge -\partial_u^2 + \left( H_u - \frac{\lambda}{2u(1+2u)^{\frac{\alpha}{2}}} - V_\chi^{uv} \right)_-
	\ge -\partial_u^2 + \left( H_u - \frac{\lambda}{u^{1+\frac{\alpha}{2}}} - V_\chi^{uv} \right)_-,
\end{equation}
where $(A)_-$ denotes the spectral projection on the negative part of the spectrum of $A$.

First, let us make a rough estimate of the spectrum of $H_u$ to
prove that this spectral projection is one-dimensional 
when $u$ is sufficiently large.
Splitting $H_u$ into three regions (again using a partition of unity), 
$v < -\delta u$, $-u < v < u$, resp. $v > \delta u$,
with a fixed $0 < \delta < 1$ and Dirichlet boundary conditions 
at $v=\pm u$ resp. $v=\pm \delta u$, we find
(with some constant $c_2 \ge (1-\delta)^{-2}$)
$$
	H_u\big|_{|v|<u} 
	\ge -\partial_v^2 + \frac{v^2}{2(u^2+u^2)^{\frac{1}{2}}} - \frac{1}{\sqrt{2u}} - \frac{c_2}{u^2}
	  = -\partial_v^2 + \frac{1}{2^{\frac{3}{2}}u}v^2 - \frac{1}{\sqrt{2u}} - \frac{c_2}{u^2},
$$
whose spectrum is bounded below by
$\left\{ \frac{1}{\sqrt{2^{3/2}u}}(2k+1) - \frac{1}{\sqrt{2u}} - \frac{c_2}{u^2} \right\}_{k=0,1,2,\ldots}$,
while
$$
	H_u\big|_{v > \delta u}
	\ge -\partial_v^2 + \frac{v^2}{2(\delta^{-2}v^2+v^2)^{\frac{1}{2}}} - \frac{1}{\sqrt{2u}} - \frac{c_2}{u^2}
	\ge \frac{\delta^2 u}{2\sqrt{2}} - \frac{1}{\sqrt{2u}} - \frac{c_2}{u^2},
$$
and similarly for $H_u|_{v < -\delta u}$, so that
\begin{eqnarray*}
	N\left(H_u - \frac{\lambda}{u^{1+\frac{\alpha}{2}}}\right)
	&\le& N\left(H_u\big|_{v<-\delta u} - \frac{\lambda}{u^{1+\frac{\alpha}{2}}}\right) 
	+ N\left(H_u\big|_{|v|<u} - \frac{\lambda}{u^{1+\frac{\alpha}{2}}}\right) \\
	&& +\ N\left(H_u\big|_{v> \delta u} - \frac{\lambda}{u^{1+\frac{\alpha}{2}}}\right)
	\le 0+1+0
\end{eqnarray*}
for $\alpha \ge 2$ and $u$ sufficiently large.
Hence, only the ground state energy of $H_u$ contributes to \eqref{H_uv_estimate}
when $M$ is taken sufficiently large,
e.g. $M \ge (c_1 + c_2 + \lambda)^{\frac{2}{3}}$ 
with $\delta = 0.8$.

A sufficient bound for the ground state energy is provided by
Proposition \ref{prop_H_hat_bound} below, showing that
$H_u \ge -\frac{1}{4u^2}$ for all $u > 0$, so that
\eqref{H_uv_estimate} becomes
$$
	H^{uv}_\lambda 
	\ge \left( -\partial_u^2 - \frac{1}{4u^2} \right) \otimes 1_{\mathpzc{h}} 
	- \left( \frac{\lambda}{u^{1+\frac{\alpha}{2}}} + V_\chi^{uv} \right) \otimes P_0,
$$
where $P_0$ denotes the projection onto the ground state of $H_u$.
Applying Lemma \ref{lem_integral_crit_bound}, we find
(after extending trivially to $[1,\infty)$)
\begin{eqnarray}
	N(H^{uv}_\lambda) 
	&\le& 8\pi C_3 \int_M^{\infty} \left( \frac{\lambda}{u^{1+\frac{\alpha}{2}}} + V_\chi^{uv} \right)^{\frac{3}{2}} u^2 (\ln u)^2 \thinspace du \nonumber \\
	&\le& 
		c_3 \left(\lambda M^{1 - \frac{\alpha}{2}} + c_1\right)^{\frac{3}{2}} (\ln \kappa^2M)^2
		+ 8\pi C_3 \lambda^{\frac{3}{2}} \int_{\kappa^2 M}^{\infty} u^{\frac{1}{2}-\frac{3\alpha}{4}} (\ln u)^2 \thinspace du,
		\nonumber \\ \label{H_uv_bound}
\end{eqnarray}
which is finite for $\alpha>2$.
This implies that $N(H_\lambda^{\mathcal{B}_1}) < \infty$ 
and, by reflection symmetry, also $N(H_\lambda^{\mathcal{B}}) < \infty$
for all $\lambda>0$.

It remains to prove that $N(H_\lambda^{\mathcal{A}}) < \infty$.
Taking the scalar lower bound for the potential of $H_\lambda^{\mathcal{A}}$,
$$
	V^\mathcal{A} := x^2y^2 - \sqrt{x^2 + y^2} - \lambda(1 + x^2 + y^2)^{-\frac{\alpha}{2}} - V_\chi,
$$
and using that
$$
	V_\chi = (x^2 + y^2)V_\chi^{uv} \le 2\sqrt{u^2 + v^2} c_1M^{-2},
$$
we have on the region $\kappa\mathcal{A}$
\begin{eqnarray}
	H_\lambda^{\mathcal{A}} 
	&\ge& -\Delta_{xy} + v^2 - \sqrt{2}(u^2+v^2)^{\frac{1}{4}} - \frac{\lambda}{(1 + 2\sqrt{u^2+v^2})^{\frac{\alpha}{2}}} - V_\chi \nonumber \\
	&\ge& -\Delta_{xy} + v^2 - \sqrt{2}(\kappa^4M^2+v^2)^{\frac{1}{4}} - \lambda - 2\sqrt{\kappa^4M^2 + v^2} c_1M^{-2},
	\label{H_A_operator_bound}
\end{eqnarray}
and since the potential of the Schr\"odinger operator on the right hand side
tends to infinity as $|\boldsymbol{x}| \to \infty \Rightarrow |v| \to \infty$,
it follows that the spectrum of $H_\lambda^{\mathcal{A}}$ is purely discrete,
and $N(H_\lambda^{\mathcal{A}}) < \infty$. 

We have proved
the first statement of Theorem \ref{thm_H_lambda}.
The second statement follows from \eqref{H_uv_bound} 
with $M = (c_1 + c_2 + \lambda)^{\frac{2}{3}}$,
together with the following bound for scalar Schr\"odinger operators
in two dimensions 
(see e.g. Theorem 20, Chapter 8.4 in \cite{Egorov-Kondratev}):
\begin{equation} \label{two-dimensional_bound}
	N(-\Delta + V) \le 1 + C_q \int_{\mathbb{R}^2} |V(\boldsymbol{x})_-|^q \left(1 + |\ln |\boldsymbol{x}||\right)^{2q-1} |\boldsymbol{x}|^{2(q-1)} \thinspace d\boldsymbol{x},
\end{equation}
with $q > 1$ and $C_q$ a positive constant\footnote{Whether 
the bound \eqref{two-dimensional_bound} extends to $q=1$ is currently unknown 
(we note that there is an error in \cite{Egorov-Kondratev-AMS}).}.
Extending $V^\mathcal{A}$ by zero outside $\kappa \mathcal{A}$, it follows that
$$
	N(H_\lambda^{\mathcal{A}}) \le 2 + 2C_q \int_{\kappa \mathcal{A}} |V^\mathcal{A}(\boldsymbol{x})_-|^q \left(1 + |\ln |\boldsymbol{x}||\right)^{2q-1} |\boldsymbol{x}|^{2(q-1)} \thinspace d\boldsymbol{x}.
$$
For large $\lambda$, we have on the unbounded region
$|v| \ge \kappa^2M$ (similarly to \eqref{H_A_operator_bound}) that
$$
	V^\mathcal{A} \ge v^2 - \sqrt{2}(2v^2)^{\frac{1}{4}} - \lambda(1 + 2|v|)^{-\frac{\alpha}{2}} - 2(2v^2)^{\frac{1}{2}} c_1M^{-2} \ge 0.
$$
Hence, the integral reduces to the bounded region 
$|u|,|v| < \kappa^2M$, i.e.
\begin{eqnarray*}
	N(H_\lambda^{\mathcal{A}}) 
	&\lesssim& \int_{x^2+y^2 < 2^{\frac{3}{2}} \kappa^2 M} (-V^\mathcal{A})^q_+ \left(1 + |\ln \sqrt{x^2 + y^2}|\right)^{2q-1} (x^2 + y^2)^{q-1} \thinspace dxdy \\
	&\le& \int_{0}^{c_4\lambda^{\frac{1}{3}}} \int_{-\pi}^{\pi}
		\left( -\frac{r^4}{4} \sin^2 2\varphi + r + \lambda(1 + r^2)^{-\frac{\alpha}{2}} + r^2 c_1M^{-2} \right)^q_+ \\
	&&\hspace{2.5cm}	\cdot \left(1 + |\ln r|\right)^{2q-1} r^{2q-1} \thinspace dr d\varphi,
\end{eqnarray*}
where we switched to polar coordinates $(r,\varphi)$.
Furthermore, this region increases in size with
$\lambda$ at a faster rate than the geometry of the potential valleys,
so we can split the integral into a central part 
and four narrowing regions along the valleys (see Figure \ref{fig_valleys}).
\begin{figure}[t]
	\centering
	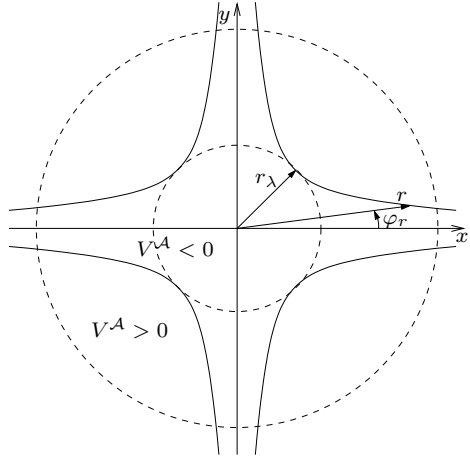
	\caption{The bounded region where $V^\mathcal{A}<0$.}
	\label{fig_valleys}
\end{figure}
We obtain the bound
\begin{eqnarray*}
	\lefteqn{ 2\pi \int_{0}^{r_\lambda}
		\left( r + \lambda(1 + r^2)^{-\frac{\alpha}{2}} + r^2 c_1\lambda^{-\frac{4}{3}} \right)^q \left(1 + |\ln r|\right)^{2q-1} r^{2q-1} \thinspace dr }\\
	&& +\ 8 \int_{r_\lambda}^{c_4\lambda^{\frac{1}{3}}} \int_{0}^{\varphi_r} 
		\left( r + \lambda(1 + r^2)^{-\frac{\alpha}{2}} + r^2 c_1\lambda^{-\frac{4}{3}} \right)^q (1 + \ln r)^{2q-1} r^{2q-1} \thinspace dr d\varphi,
\end{eqnarray*}
where $r_\lambda$ is the solution to
$$
	-\frac{1}{4}r^4 + r + \lambda(1 + r^2)^{-\frac{\alpha}{2}} + r^2 c_1M^{-2} = 0,
$$
i.e. $r_\lambda \sim \lambda^{\frac{1}{4+\alpha}}$,
and $\varphi_r \sim \frac{1}{r^2}(r + \lambda r^{-\alpha})^{\frac{1}{2}}$.
The first integral is bounded by
\begin{equation} \label{first_integral_bound}
	r_\lambda^{2q} \left(r_\lambda + \lambda + r_\lambda^2 c_1 \lambda^{-\frac{4}{3}}\right)^q (1 + \ln r_\lambda)^{2q-1}
	\le c_5 \lambda^{q\frac{6+\alpha}{4+\alpha}} (\ln \lambda)^{2q-1},
\end{equation}
and the second by
\begin{equation} \label{second_integral_bound}
	c_6 \int_{c_7\lambda^{\frac{1}{4+\alpha}}}^{c_4\lambda^{\frac{1}{3}}} 
		\frac{1}{r^2} ( \underbrace{r + \lambda r^{-\alpha} + r^2 c_1 \lambda^{-\frac{4}{3}}}_{ \le 3\lambda^{\frac{2}{3}} } )^{q + \frac{1}{2}} (r\ln r)^{2q-1} \thinspace dr
	\le \frac{c_8}{q-1} \lambda^{\frac{4}{3}q - \frac{1}{3}} (\ln \lambda)^{2q-1}.
\end{equation}
Now, for $\alpha > 2$ we can choose $q$ sufficiently close to $1$ to make
these expressions dominated by $o(\lambda^{\frac{4}{3}})$.
On the other hand, the first term of the r.h.s. of 
\eqref{H_uv_bound} is asymptotically bounded by
$$
	c_9 \left( \lambda \cdot (\lambda^{2/3})^{1 - \frac{\alpha}{2}} \right)^{\frac{3}{2}} (\ln \lambda)^2
	= c_9 \lambda^{\frac{3}{2} - \frac{1}{2}(\alpha-2)} (\ln \lambda)^2
$$
for $2 < \alpha \le 5$ and by $c_9(\ln \lambda)^2$ otherwise,
and for the second term we have, 
with $\alpha =: 2 + \frac{4}{3}a$ and any $0 \le \epsilon < 1$, 
$$
	\int_{M}^{\infty} u^{-1 - a} (\ln u)^2 \thinspace du
	\le M^{-\epsilon a} \int_{1}^{\infty} u^{-1 - (1-\epsilon)a} (\ln u)^2 \thinspace du
	\le \lambda^{-\frac{2}{3}\epsilon a} \cdot \frac{2}{((1-\epsilon) a)^3}.
$$
Summing up, we obtain
$$
	N(H_\lambda) \le C(\alpha)
		+ 32\pi C_3 \lambda^{\frac{3}{2} - \epsilon(\alpha)} \cdot \frac{128}{27(\alpha-2)^3} \qquad \forall \lambda>0,
$$
for some constant $C(\alpha)$ and sufficiently small $\epsilon(\alpha) > 0$. 
\qed

\vspace{\baselineskip}
It follows from Theorem \ref{thm_H_lambda} that the 
asymptotic eigenvalue distribution of
the weighted Hamiltonian $\tilde{H}$ is given by
$$
	N(\tilde{H} - \lambda) \sim o(\lambda^{\frac{3}{2}}), \quad \lambda \to \infty,
$$
regardless of $\alpha > 2$.
We note that the same approach can be applied to the
purely bosonic model, i.e. the scalar Schr\"odinger operator 
$H_B = -\Delta + x^2y^2$, with $\alpha \ge 0$.
In this case there will be no contribution from the region $\mathcal{B}$ 
when $M \sim \lambda^2$, and the correct leading order eigenvalue asymptotics
for $\alpha = 0$ (see \cite{Simon_asymptotics}),
$$
	N(H_B - \lambda) \sim \lambda^{\frac{3}{2}} \ln \lambda, \quad \lambda \to \infty,
$$
would be matched by the corresponding bound \eqref{first_integral_bound} for the central region with $q=1$,
while for the cut off valleys there is a bound analogous to 
\eqref{second_integral_bound} with 
$$
	\int_{\lambda^{\frac{1}{4+\alpha}}}^{\lambda} 
		(\underbrace{\lambda r^{-\alpha} + r^2 c_1 \lambda^{-4}}_{ \le 2\lambda^{\frac{4}{4+\alpha}} })^{q + \frac{1}{2}} r^{2q-3} (\ln r)^{2q-1} \thinspace dr
	\le \frac{c_{10}}{q-1} \lambda^{\frac{4q+2}{4+\alpha} + 2(q-1)} (\ln \lambda)^{2q-1}.
$$
One could try to improve this by instead letting $M$ be fixed and 
reconsidering the bound on the region $\mathcal{B}_1$.
In any case, we have for a nonzero weight that 
$N(\tilde{H}_B - \lambda) \sim o(\lambda^{\frac{3}{2}})$, $\lambda \to \infty$.

\subsubsection{Asymptotics of $H_u$}

We conclude this section with some useful properties of the 
operator $H_u$ in the limit $u \to \infty$ \cite{Graf}. 
By the change of variable $v = u^{\frac{1}{4}}t$,
we write $H_u = u^{-\frac{1}{2}} \hat{H}(u^{-\frac{3}{2}})$, with
\begin{equation} \label{H_hat}
	\hat{H}(\epsilon) := -\partial_t^2 
		+ \frac{t^2}{2(1 + \epsilon t^2)^{\frac{1}{2}}}
		- \frac{1}{\sqrt{2}(1 + \epsilon t^2)^{\frac{1}{4}}}.
\end{equation}

\begin{prop} \label{prop_H_hat_bound}
	$\hat{H}(\epsilon) \ge -\frac{\epsilon}{4}$, \ 
	for all $\epsilon > 0$.
\end{prop}
\begin{proof}
	We use that for any $f = f(t)$
	$$
		(-i\partial_t + if)(-i\partial_t - if) \ge 0,
	$$
	i.e. $-\partial_t^2 + f^2 - f' \ge 0$.
	As a first attempt, let
	$$
		f_0 := \frac{t}{\sqrt{2}(1 + \epsilon t^2)^{\frac{1}{4}}},
	$$
	resulting in
	\begin{equation} \label{H_hat_first_bound}
		\hat{H}(\epsilon) \ge -\frac{\epsilon t^2}{2\sqrt{2}(1 + \epsilon t^2)^{\frac{5}{4}}}.
	\end{equation}
	While the r.h.s. is bounded and vanishes as $\epsilon \to 0$
	pointwise, it does not so uniformly. Consider instead
	$f = f_0 + \epsilon f_1$,
	with
	$$
		f_1 := -\frac{t}{4(1 + \epsilon t^2)}.
	$$
	We so get the bound \eqref{H_hat_first_bound} pushed to $O(\epsilon)$:
	\begin{eqnarray*}
		\hat{H}(\epsilon) &\ge&
		-\frac{\epsilon t^2}{2\sqrt{2}(1 + \epsilon t^2)^{\frac{5}{4}}} 
		- 2\epsilon f_0f_1 - \epsilon^2f_1^2 + \epsilon f_1' \\
		&=& - \epsilon^2f_1^2 + \epsilon f_1'
		= -\frac{\epsilon}{4} \cdot \frac{1 - \frac{3}{4} \epsilon t^2}{(1 + \epsilon t^2)^2},
		\ge -\frac{\epsilon}{4}. 
	\end{eqnarray*}
\end{proof}

Let $\hat{P}_0$ denote the projection onto the ground state of
$\hat{H}_0 := \hat{H}(0)$, i.e. 
$$
	(P_0 \psi)(t) = \varphi_0(t) \int \overline{\varphi_0(\tau)} \psi(\tau) d\tau,
$$
where $\varphi_0(t) = (\sqrt{2}\pi)^{-\frac{1}{4}} e^{-t^2/(2\sqrt{2})}$ 
is its normalized wave function,
and let $\hat{P}_0^\perp := 1 - \hat{P}_0$.
Note that
$$
	\hat{H}(\epsilon) - \hat{H}_0 
	= -\frac{t^2}{2}\left( 1 - (1+\epsilon t^2)^{-\frac{1}{2}} \right) 
	+ \frac{1}{\sqrt{2}}\left( 1 - (1+\epsilon t^2)^{-\frac{1}{4}} \right)
$$
and
\begin{equation} \label{phi_exp_value}
	\langle \varphi_0, \hat{H}(\epsilon) \varphi_0 \rangle
	= -\frac{\epsilon}{4} + o(\epsilon).
\end{equation}

\begin{prop} \label{prop_H_hat_proj_bound}
	$\hat{H}(\epsilon) \ge \left( -\frac{\epsilon}{4} + o(\epsilon) \right) \hat{P}_0 + c\hat{P}_0^\perp,$ \ 
	as $\epsilon \to 0$, where $c>0$.
\end{prop}

We start with

\begin{lem} \label{lem_H_hat_perp_bound}
	For small $\epsilon > 0$,
	$$
		\hat{P}_0^\perp \hat{H}(\epsilon) \hat{P}_0^\perp \ \ge \ \frac{\sqrt{2}}{2} \hat{P}_0^\perp.
	$$
	(Note: $\sqrt{2}$ is the excitation energy of $\hat{H}_0$.)
\end{lem}
\begin{proof}
	We again use a partition of unity and let
	$\tilde{f}_i = \tilde{f}_i(s)$, $(i=1,2)$ be smooth functions with
	$\tilde{f}_1^2 + \tilde{f}_2^2 = 1$, $\tilde{f}_2(s) = 0$ for $|s| \le 1$,
	and $\tilde{f}_1(s) = 0$ for $|s| \ge 2$.
	Set $f_i(t) = \tilde{f}_i(t/R)$.
	Then
	$$
		\hat{H}(\epsilon) = f_1 \hat{H}(\epsilon) f_1 + f_2 \hat{H}(\epsilon) f_2 + O(R^{-2}).
	$$
	For large $R$, that error is $\le \sqrt{2}/10$ and
	$$
		\|[f_1,\partial_t^2]\hat{P}_0\|, \qquad \|(1-f_1)\hat{P}_0\|
	$$
	have the same bound (for later use).
	The potential of $\hat{H}(\epsilon)$ in \eqref{H_hat}, 
	denote it $V(\epsilon,t)$, satisfies
	$$
		V(\epsilon,t) \ge \frac{t^2}{2(1 + t^2)^{\frac{1}{2}}} - \frac{1}{\sqrt{2}}
	$$
	for $0<\epsilon<1$. This is $\ge \sqrt{2}$ for 
	$t \in \supp f_2$ if $R$ is large enough.
	Hence, we obtain
	$f_2 \hat{H}(\epsilon) f_2 \ge \sqrt{2} f_2^2$.
	Now, for fixed $R$,
	$$
		f_1 \hat{H}(\epsilon) f_1 = f_1 \hat{H}_0 f_1 + O(\epsilon),
	$$
	and we take $\epsilon$ small enough that 
	$|O(\epsilon)| \le \sqrt{2}/10$. We consider
	\begin{eqnarray}
		\hat{P}_0^\perp f_1 \hat{H}_0 f_1 \hat{P}_0^\perp
		&=& f_1 \hat{P}_0^\perp \hat{H}_0 \hat{P}_0^\perp f_1
		+ (\hat{P}_0^\perp f_1 - f_1 \hat{P}_0^\perp) \hat{H}_0 f_1 \hat{P}_0^\perp \nonumber \\
		&& +\ f_1 \hat{P}_0^\perp \hat{H}_0 (f_1 \hat{P}_0^\perp - \hat{P}_0^\perp f_1). \label{proj_tricks}
	\end{eqnarray}
	Using $\hat{H}_0 \hat{P}_0 = 0$ we have
	$$
		\hat{H}_0 (f_1 \hat{P}_0^\perp - \hat{P}_0^\perp f_1) = \hat{H}_0 (\hat{P}_0 f_1 - f_1 \hat{P}_0) 
		= (f_1 \hat{H}_0 - \hat{H}_0 f_1) \hat{P}_0 = [f_1,-\partial_t^2] \hat{P}_0
	$$
	for the last term of \eqref{proj_tricks}, and similarly for the second.
	Together with the bound 
	$\hat{P}_0^\perp \hat{H}_0 \hat{P}_0^\perp \ge \sqrt{2} \hat{P}_0^\perp$
	we conclude
	$$
		\hat{P}_0^\perp \hat{H}(\epsilon) \hat{P}_0^\perp 
		\ge \sqrt{2}(f_1 \hat{P}_0^\perp f_1 + f_2^2) - \frac{\sqrt{2}}{10}(1+1+2).
	$$
	Multiplying again with $\hat{P}_0^\perp$ and using
	$$
		\hat{P}_0^\perp f_1 \hat{P}_0^\perp f_1 \hat{P}_0^\perp
		= \hat{P}_0^\perp f_1^2 \hat{P}_0^\perp - \hat{P}_0^\perp f_1 \hat{P}_0 f_1 \hat{P}_0^\perp,
	$$
	and $-\hat{P}_0^\perp f_1 \hat{P}_0 = \hat{P}_0^\perp (1 - f_1) \hat{P}_0$,
	we obtain
	\begin{equation*}
		\hat{P}_0^\perp \hat{H}(\epsilon) \hat{P}_0^\perp
		\ge \sqrt{2} \hat{P}_0^\perp (\underbrace{f_1^2 + f_2^2}_{=1}) \hat{P}_0^\perp - \frac{\sqrt{2}}{10}(4+1)\hat{P}_0^\perp
		= \frac{\sqrt{2}}{2} \hat{P}_0^\perp. \qedhere
	\end{equation*}
\end{proof}

\begin{proof}[Proof of Proposition \ref{prop_H_hat_proj_bound}]
	We decompose
	$$
		\hat{H}(\epsilon) = \hat{P}_0 \hat{H}(\epsilon) \hat{P}_0 
		+ \hat{P}_0^\perp \hat{H}(\epsilon) \hat{P}_0^\perp
		+ \hat{P}_0^\perp (\hat{H}(\epsilon) - \hat{H}_0) \hat{P}_0
		+ \hat{P}_0 (\hat{H}(\epsilon) - \hat{H}_0) \hat{P}_0^\perp,
	$$
	since $\hat{P}_0^\perp \hat{H}_0 \hat{P}_0 = 0$.
	The first two terms are greater than
	$-\frac{\epsilon}{4} \hat{P}_0 + o(\epsilon)$ by \eqref{phi_exp_value},
	resp. $\frac{1}{\sqrt{2}} \hat{P}_0^\perp$ by Lemma \ref{lem_H_hat_perp_bound}.
	Expectations of the third one are bounded as
	\begin{eqnarray*}
		|\langle \psi, \hat{P}_0^\perp (\hat{H}(\epsilon) - \hat{H}_0) \hat{P}_0 \psi \rangle |
		&\le& \|\hat{P}_0^\perp \psi\| \|(\hat{H}(\epsilon) - \hat{H}_0)\hat{P}_0 \psi\|
		 \le  c\epsilon \|\hat{P}_0^\perp \psi\| \|\hat{P}_0 \psi\| \\
		&\le& c\epsilon \left( \epsilon^{-\frac{1}{2}} \|\hat{P}_0^\perp \psi\|^2 + \epsilon^{\frac{1}{2}} \|\hat{P}_0 \psi\|^2 \right),
	\end{eqnarray*}
	and so for the fourth one. Therefore,
	$$
		\hat{H}(\epsilon) \ge \left( -\frac{\epsilon}{4} + o(\epsilon) \right) \hat{P}_0 
		+ \left( \frac{1}{\sqrt{2}} - c'\epsilon^{\frac{1}{2}} \right) \hat{P}_0^\perp,
	$$
	where the second bracket is positive for $\epsilon$ small enough.
\end{proof}

\section{CLR bound for operator-valued potentials}

Given a separable Hilbert space $\mathpzc{h}$, we denote by $S^p(\mathpzc{h})$
the set of compact symmetric operators $A$ on $\mathpzc{h}$ s.t. 
$\tr_{\mathpzc{h}} |A|^p = \tr_{\mathpzc{h}} (A^\dagger A)^{p/2} < \infty$.
The following theorem is given as Corollary 2.4 in 
\cite{Hundertmark} (see also \cite{Frank_et_al}):

\begin{thm} \label{thm_Hundertmark}
	Let $\mathpzc{h}$ be some auxiliary Hilbert space and
	$V$ a potential in $L^{d/2}(\mathbb{R}^d; S^{d/2}(\mathpzc{h}))$, $d \ge 3$.
	Then
	$$
		N(-\Delta \otimes 1_{\mathpzc{h}} + V) 
		\le C_d \int_{\mathbb{R}^d} \tr_{\mathpzc{h}} \left| V(x)_- \right|^{\frac{d}{2}} \thinspace dx
	$$
	for some positive constant $C_d$.
\end{thm}

\noindent
It is also noted in \cite{Hundertmark} that the operator
$H_d := -\Delta \otimes 1_{\mathpzc{h}} + V$ is self-adjoint 
and semi-bounded from below
on the corresponding Sobolev space $H_1(\mathbb{R}^d;\mathpzc{h})$,
and that $\sigma_{\textrm{ess}}(H_d) \subseteq [0,\infty)$.

From the above theorem we can derive the following \cite{Safronov}:

\begin{lem} \label{lem_integral_bound}
	Assume $\mathpzc{h}$ is an auxiliary Hilbert space
	and $V: [0,\infty) \to S^{\frac{3}{2}}(\mathpzc{h})$ 
	a smooth operator-valued potential.
	Let $H_1 := -\partial_x^2 \otimes 1_{\mathpzc{h}} + V$
	be self-adjoint and defined by Friedrichs extension on $C_0^\infty(\mathbb{R}_+;\mathpzc{h})$.
	Then
	$$
		N(H_1) 
		\le 4\pi C_3 \int_0^\infty \tr_{\mathpzc{h}} \left| V(x)_- \right|^{\frac{3}{2}} x^2 \thinspace dx.
	$$
\end{lem}
\begin{proof}
	Consider $N(H_1) = \sup_{W \in \mathcal{W}_1} \dim W$,
	where $\mathcal{W}_1$ denotes the set of linear subspaces 
	$W \subseteq C_0^\infty(\mathbb{R}_+;\mathpzc{h}) \subseteq L^2(\mathbb{R}_+;\mathpzc{h})$
	s.t.
	$
		\langle u, H_1 u \rangle < 0 \  \forall u \in W.
	$
	For $u \in W \in \mathcal{W}_1$ and $x \in \mathbb{R}^3$,
	$r := |x|$, we let $\psi(x) := \frac{1}{r}u(r)$. Then
	$\psi \in C_0^\infty(\mathbb{R}^3;\mathpzc{h})$ and
	\begin{eqnarray*}
		\langle \psi, H_3 \psi \rangle 
		&=& \int_{\mathbb{R}^3} \langle \psi, \left( -\Delta_{\mathbb{R}^3} \otimes 1_\mathpzc{h} + V(|x|) \right) \psi \rangle_\mathpzc{h} \thinspace dx \\
		&=& |S^2| \int_0^\infty \left\langle \frac{1}{r}u(r), 
			\left( -\frac{1}{r} \frac{\partial^2}{\partial r^2} r \otimes 1_\mathpzc{h} + V(r) \right) \frac{1}{r}u(r) \right\rangle_\mathpzc{h} r^2 \thinspace dr \\
		&=& 4\pi \int_0^\infty \left\langle u(r), 
			\left( -\partial_r^2 \otimes 1_\mathpzc{h} + V(r) \right) u(r) \right\rangle_\mathpzc{h} dr \\
		&=& 4\pi \langle u, H_1 u \rangle < 0.
	\end{eqnarray*}
	Hence, $\psi \in W' \in \mathcal{W}_3$ for some $W'$,
	where $\mathcal{W}_3$ denotes the corresponding set of linear subspaces 
	$W' \subseteq C_0^\infty(\mathbb{R}^3;\mathpzc{h}) \subseteq L^2(\mathbb{R}^3;\mathpzc{h})$
	s.t.
	$
		\langle \psi, H_3 \psi \rangle < 0 \  \forall \psi \in W'.
	$
	Also, if $u_1,u_2$ in $W \in \mathcal{W}_1$ are orthogonal, then
	so are the associated $\psi_1, \psi_2$ in $W' \in \mathcal{W}_3$,
	so that to each $W \in \mathcal{W}_1$ there corresponds a $W' \in \mathcal{W}_3$
	with $\dim W' \ge \dim W$. Hence,
	\begin{eqnarray*}
		N(H_1) &=& \sup_{W \in \mathcal{W}_1} \dim W 
		\le \sup_{W' \in \mathcal{W}_3} \dim W' = N(H_3) \\
		&\le& C_3 \int_{\mathbb{R}^3} \tr_{\mathpzc{h}} \left| V(|x|)_- \right|^{\frac{3}{2}} \thinspace dx 
		= 4\pi C_3 \int_0^\infty \tr_{\mathpzc{h}} \left| V(r)_- \right|^{\frac{3}{2}} r^2 \thinspace dr,
	\end{eqnarray*}
	by Theorem \ref{thm_Hundertmark}.
\end{proof}

\begin{lem} \label{lem_integral_crit_bound}
	With $\mathpzc{h}$ and $V$ as above,
	let $H_2 := (-\partial_x^2 - \frac{1}{4x^2}) \otimes 1_{\mathpzc{h}} + V$
	be self-adjoint and defined by Friedrichs extension on $C_0^\infty((1,\infty);\mathpzc{h})$.
	Then
	$$
		N(H_2) 
		\le 4\pi C_3 \int_1^\infty \tr_{\mathpzc{h}} \left| V(x)_- \right|^{\frac{3}{2}} x^2 (\ln x)^2 \thinspace dx.
	$$
\end{lem}
\begin{proof}
	We have for $u \in C_0^\infty((1,\infty);\mathpzc{h})$ that
	$$
		\langle u, H_2 u \rangle
		= \int_1^\infty \left( 
			\|u'(x)\|_{\mathpzc{h}}^2 - \frac{1}{4x^2} \|u(x)\|_{\mathpzc{h}}^2 + \langle u(x),V(x)u(x) \rangle_{\mathpzc{h}}
			\right) dx.
	$$
	Note that
	\begin{eqnarray*}
		\left\| \left({ \textstyle \partial_x - \frac{1}{2x} }\right)u(x) \right\|_\mathpzc{h}^2
		&=& \|u'(x)\|_\mathpzc{h}^2 
			- \frac{1}{2x}\left( \langle u'(x),u(x) \rangle_\mathpzc{h} + \langle u(x),u'(x) \rangle_\mathpzc{h} \right) \\
		&&	+\ \frac{1}{4x^2}\|u(x)\|_\mathpzc{h}^2,
	\end{eqnarray*}
	so that after integrating by parts,
	$$
		\langle u, H_2 u \rangle
		= \int_1^\infty \left( 
			\left\| \left({ \textstyle \partial_x - \frac{1}{2x} }\right)u(x) \right\|_\mathpzc{h}^2 
				+ \langle u(x),V(x)u(x) \rangle_{\mathpzc{h}}
			\right) dx.
	$$
	We can write $u(x) = x^{\frac{1}{2}}v(x)$, 
	with $v \in C_0^\infty((1,\infty);\mathpzc{h})$, implying
	$$
		\langle u, H_2 u \rangle
		= \int_1^\infty \left( 
			\| x^{\frac{1}{2}}v'(x) \|_\mathpzc{h}^2 + \langle v(x),xV(x)v(x) \rangle_{\mathpzc{h}}
			\right) dx.
	$$
	Put $t := \ln x$ and $w(t) := v(e^t)$. 
	Then $w \in C_0^\infty((0,\infty);\mathpzc{h})$ and
	\begin{eqnarray*}
		\langle u, H_2 u \rangle
		&=& \int_0^\infty \left( 
			\|w'(t)\|_\mathpzc{h}^2 + \langle w(t),e^{2t}V(e^t)w(t) \rangle_{\mathpzc{h}}
			\right) dt \\
		&=& \int_0^\infty \left\langle w(t), 
			\left( -\partial_t^2 \otimes 1_\mathpzc{h} + e^{2t}V(e^t) \right)w(t) 
			\right\rangle_{\mathpzc{h}} dt.
	\end{eqnarray*}
	We also note that there is
	a 1-to-1 correspondence between linearly independent
	sets of such $u \in C_0^\infty((1,\infty);\mathpzc{h})$ 
	and $w \in C_0^\infty((0,\infty);\mathpzc{h})$.
	Applying Lemma \ref{lem_integral_bound} with the potential
	$W(x) = e^{2x}V(e^x)$ we find
	\begin{eqnarray*}
		N(H_2) &=& N\left( -\partial_x^2 \otimes 1_\mathpzc{h} + e^{2x}V(e^x) \right) \\
		&\le& 4\pi C_3 \int_0^\infty \tr_{\mathpzc{h}} \left| \left( e^{2x}V(e^x) \right)_- \right|^{\frac{3}{2}} x^2 \thinspace dx \\
		&=&   4\pi C_3 \int_1^\infty \tr_{\mathpzc{h}} \left| V(s)_- \right|^{\frac{3}{2}} s^2 (\ln s)^2 \thinspace ds,
	\end{eqnarray*}
	where we substituted $s := e^x$.
\end{proof}

\subsubsection*{Acknowledgements}

I would like to express my sincere thanks to Oleg Safronov
for initiating this approach and pointing out the CLR bound
of Lemma \ref{lem_integral_crit_bound}.
I would also like to thank Gian Michele Graf 
(in particular in connection with Propositions \ref{prop_H_hat_bound} and \ref{prop_H_hat_proj_bound}), 
Jens Hoppe, and Ari Laptev for useful discussions and valuable suggestions, 
as well as Giovanni Felder and ETH Z\"urich for hospitality.
This work was supported by the Swedish Research Council and
the European Science Foundation activity MISGAM.

\end{document}